\documentclass[% preprint, 
superscriptaddress,
notitlepage,
 amsmath,amssymb,
 aps,
]{revtex4-1}

\usepackage{graphicx}% Include figure files
\usepackage{dcolumn}% Align table columns on decimal point
\usepackage{bm}% bold math
\usepackage{amsmath,amsfonts,amssymb}
\usepackage{graphicx,amsmath, amscd}
\usepackage{color,comment}
\usepackage[T1]{fontenc}
\usepackage[latin1]{inputenc}
\usepackage{dsfont}
\usepackage{amsthm}
\usepackage{mathtools}
\usepackage{enumerate}

\usepackage[lmargin=2cm,rmargin=2cm,tmargin=2cm,bmargin=2cm]{geometry}
\usepackage[toc,page]{appendix}

\usepackage{tikz}
\usetikzlibrary{decorations.pathreplacing,angles,quotes}

\newtheorem{theorem}{Theorem}[section]
\newtheorem{lemma}[theorem]{Lemma}
\newtheorem{corollary}[theorem]{Corollary}

\newtheorem{defi/prop}[theorem]{Definition/Proposition}

\newcommand{\N}{\mathds{N}}

\newcommand{\C}{\mathds{C}}
\renewcommand{\P}{\mathds{P}}

\renewcommand{\leq}{\leqslant}
\renewcommand{\geq}{\geqslant}

\DeclareMathOperator{\tr}{Tr}

\DeclareMathOperator{\E}{\mathds{E}}
\DeclareMathOperator{\SR}{SR}
\DeclareMathOperator{\SN}{SN}

\DeclarePairedDelimiter{\ceil}{\lceil}{\rceil}
\DeclarePairedDelimiter{\floor}{\lfloor}{\rfloor}

\newcommand{\id}{\mathds{1}}

\newcommand{\ketbra}[2]{| #1 \rangle\!\langle #2 |}
\newcommand{\bra}[1]{\langle #1 |}
\newcommand{\ket}[1]{| #1 \rangle}
\newcommand*{\coloneqq}{\mathrel{\vcenter{\baselineskip0.5ex \lineskiplimit0pt \hbox{\scriptsize.}\hbox{\scriptsize.}}} =}

\newcommand{\proj}[2]{| #1 \rangle\!\langle #2 |}

\theoremstyle{plain}
\newtheorem{thm}{Theorem}[section]
\newtheorem{lem}{Lemma}[section]
\newtheorem{cor}{Corollary}[section]

\newtheorem{defn}{Definition}[section]

\def\beq{\begin{equation}}
\def\eeq{\end{equation}}
\def\bq{\begin{quote}}
\def\eq{\end{quote}}
\def\ben{\begin{enumerate}}
\def\een{\end{enumerate}}
\def\bit{\begin{itemize}}
\def\eit{\end{itemize}}

\def\ra{\rightarrow}

\def\lb{\left(}
\def\rb{\right)}
\def\lset{\lbrace}
\def\rset{\rbrace}

\def\l|{\left|}
\def\r|{\right|}

\def\ident{\textnormal{id}}
\def\one{\id}

\newcommand\M{\mathcal{M}}

\newcommand\me{\omega}

\newcommand\ME{\Omega}

\newcommand\trans{\vartheta}

\newcommand\flip{\mathbb{F}}

\newcommand\PPTSet{\mathrm{PPT}}

\begin{document}

\title{High-Dimensional Entanglement in States with Positive Partial Transposition}

\author{Marcus Huber}
\affiliation{Institute for Quantum Optics and Quantum Information, Austrian Academy of Sciences, Boltzmanngasse 3, 1090 Vienna, Austria}%

\author{Ludovico Lami}
\affiliation{Centre for the Mathematics and Theoretical Physics of Quantum Non-Equilibrium Systems, School of Mathematical Sciences, University of Nottingham, Nottingham NG7 2RD, United Kingdom}%

\author{C\'{e}cilia Lancien}
\affiliation{Departamento de An\'{a}lisis Matem\'{a}tico, Universidad Complutense de Madrid, Plaza de Ciencias 3, 28040 Madrid, Spain}%

\author{Alexander M\"{u}ller-Hermes}
\affiliation{QMATH, Department of Mathematical Sciences, University of Copenhagen, Universitetsparken 5, 2100 Copenhagen, Denmark}%

\date{\today}% It is always \today, today,

\begin{abstract}
Genuine high-dimensional entanglement, i.e.~the property of
having a high Schmidt number, constitutes a resource in quantum communication, overcoming limitations of low-dimensional systems.
States with a positive partial transpose (PPT), on the other hand,
are generally considered weakly entangled, as they can never be
distilled into pure entangled states. This naturally raises the
question, whether high Schmidt numbers are possible for PPT states.
Volume estimates suggest that optimal, i.e. linear, scaling in local dimension should be possible, albeit without providing an insight
into the possible slope. We provide the first explicit construction
of a family of PPT states that achieves linear scaling in local
dimension and we prove that random PPT states typically share this
feature. Our construction also allows us to answer a recent question by Chen et al.~on the existence of PPT states whose Schmidt number increases by an arbitrarily large amount upon partial transposition. Finally, we link the Schmidt number to entangled sub-block matrices of a quantum state.  We use this connection to prove that quantum states invariant under partial transposition on the smaller of their two subsystems cannot have maximal Schmidt number. This generalizes a well-known result by Kraus et al. We also show that the Schmidt number of absolutely PPT states cannot be maximal, contributing to an open problem in entanglement theory.
\end{abstract}

\maketitle

\section{Introduction}

Entanglement is a cornerstone of quantum information theory and an important resource for quantum and private communication~\cite{book}. Current communication experiments are often based on entanglement between two degrees of freedom~\cite{sota}, i.e.~\emph{qubits}. This is despite the fact that through recent technological advances  a growing number of high-dimensional quantum systems can be controlled~\cite{exp1,exp2,exp3,exp4,exp5}.  Using such high-dimensional systems can for instance increase the resistance to noise~\cite{hd1,hd2,hd3,hd4,hd5} compared to low-dimensional implementations. However, for such improvements genuine high-dimensional entanglement is needed, as opposed to merely entanglement in a high-dimensional system. One natural measure to certify this dimensionality of entanglement is given by the Schmidt number~\cite{Terhal00}: This number is easily understood (a more formal definition will follow later) as the minimal local dimension of entangled systems needed to reproduce a specific state via local operations and classical communication. In other words, certifying a high Schmidt number shows that the state could not have arisen from low-dimensional quantum systems. 

%One could also argue that multiple copies of qubits constitute the same resource as genuine high-dimensional systems. While for local systems this is of course true, assuming one can implement any unitary on that Hilbert space, this equivalence fails for entangled states intended for communication. Apart from physical noise models generically depending on the particles, rather than their local Hilbert space dimensions, being restricted by locality adds additional constraints that imply the existence of qudit correlations that are impossible to achieve with multiple qubits \cite{ottinico}.

Another possible way to certify the dimensionality of entanglement in high-dimensional systems is given by the positive partial transpose (PPT) criterion~\cite{EltschkaSiewert}. Unfortunately, most entangled states have a positive partial transpose~\cite{AS} and the PPT criterion cannot be used to certify the dimensionality of this entanglement. This gives rise to a natural question: \emph{Can high-dimensional entanglement occur at all in systems that are noisy enough to be PPT?} By letting the local dimensions grow, examples of PPT states with increasingly high Schmidt number have been constructed in~\cite{Chen17}. However, their Schmidt number scales only logarithmically in the local dimension. It also turns out that for systems of local dimension $d=3$ the maximum Schmidt number cannot be reached by states which have a positive partial transpose~\cite{Yang16} (this had been earlier conjectured in~\cite{Sanpera01}). 

In our article we study the above questions. The following is an outline of our results: 
\begin{itemize}
    \item In Section \ref{sec explicit} we \emph{explicitly} construct a general family of PPT states with local dimension $d$ and Schmidt number scaling as $d/4$. As a byproduct, we also solve a recent conjecture~\cite[Conjecture 36]{Chen17} by exhibiting PPT states whose Schmidt number increases by an arbitrarily large amount upon partial transposition.
    \item In Section \ref{sec:HighDim} we investigate, whether Schmidt numbers scaling linear in the local dimension are \emph{generic} among PPT states. We first comment on results from~\cite[\S 4.3]{SWZ} comparing the volume of the set of PPT states with the volume of the set of states with `not too large' Schmidt number. This comparison implies that indeed most PPT states have Schmidt number scaling linearly in the local dimensions. We proceed by constructing a simple family of random states that exhibit this behaviour with, asymptotically, high probability.
    \item Finally in Section \ref{sec:UpperBounds} we refine a method from~\cite{Yang16} to study when PPT states in bipartite systems cannot reach maximal Schmidt number. For states invariant under a partial transposition on the smaller dimensional subsystem we show that their Schmidt number cannot be maximal. This generalizes a well-known result from~\cite{2xNseparability}. Finally, we find that the same conclusion is true for absolutely PPT states~\cite{Hildebrand}. This contributes to related open problems in entanglement theory~\cite{abs}. 
    
\end{itemize}

\section{Preliminaries and notation}

Let us fix some notation that we will be using throughout our article. We will denote the set of $d\times d$ matrices with complex entries by $\M_{d}$, and the subcone of positive semidefinite (commonly referred to simply as positive) matrices by $\M^+_d$. Occasionally, we will write $X\geq 0$ to abbreviate $X\in \M_{d}^+$ when the dimension of $X$ is clear from context. The identity matrix will be denoted by $\one$, and sometimes we will write $\one_d$ to emphasize its dimension. Recall that the state of a quantum system is modelled by a positive matrix with unit trace (referred to as simply a quantum state).

To study quantum entanglement, we will often consider operators acting on tensor products of Hilbert spaces. It will be convenient to use labels like $A,B$ etc. to denote different tensor factors in these situations. For example, given an operator $X\in \M_{d_A}\otimes \M_{d_B}$ we will sometimes write $X_{AB}$ where the labels $A$ and $B$ refer to the two tensor factors of the space $\C^{d_A}\otimes \C^{d_B}$ on which the operator acts. This notation will be particularly useful when considering linear maps acting partially on tensor products. Given a linear map $\mathcal{L}:\M_{d_A}\ra\M_{d_C}$ we would write $(\ident_{A}\otimes \mathcal{L}_B)(X_{AB})$ for its partial application to the operator $X_{AB}$ introduced before. Note that $\ident:\M_{d}\ra\M_d$ denotes the identity map and we will sometimes write $\ident_d$ to emphasize the dimension.

\subsection{Entanglement and Schmidt number}

A quantum state $\rho_{AB}$ on $\C^{d_A}\otimes\C^{d_B}$ is called separable if it can be written as 
\[
\rho_{AB} = \sum^m_{i=1} p_i \sigma^{(i)}_{A}\otimes \tau^{(i)}_{B}
\]
for some $m\in \N$, a probability distribution $\lset p_i\rset^m_{i=1}$, and quantum states $\sigma^{(i)}_{A}$ on $\C^{d_A}$ and $\tau^{(i)}_{B}$ on $\C^{d_B}$, $1\leq i\leq m$. Any quantum state that is not separable is called entangled. 

To quantify different degrees of entanglement various measures of entanglement have been considered. Here we will focus on one such measure called the Schmidt number,first introduced in~\cite{Terhal00}. Given a pure state $\ket{\psi}_{AB}\in \C^{d_A}\otimes \C^{d_B}$ we can define its Schmidt rank as
\begin{equation*}
    \text{SR}\lb\ket{\psi}_{AB}\rb \coloneqq \text{rk}\lb \tr_A(\proj{\psi}{\psi}_{AB})\rb \, ,
\label{equ:SR}
\end{equation*}
where $\tr_A$ denotes the partial trace over the first tensor factor and $\proj{\psi}{\psi}_{AB}$ is the rank-$1$ projector onto the span of $\ket{\psi}_{AB}$. Now, given any quantum state $\rho_{AB}$ on $\C^{d_A}\otimes\C^{d_B}$ we define its Schmidt number as 
\begin{equation*}
    \SN(\rho_{AB}) \coloneqq \min \left\{ \max_{1\leq i\leq m}\, \SR\left(\ket{\psi^{(i)}}_{AB}\right) :\ \sum_{i=1}^m p_i  \proj{\psi^{(i)}}{\psi^{(i)}}_{AB}=\rho_{AB} \right\} \, ,
\label{equ:SN}
\end{equation*}
where the minimum is over decompositions of $\rho_{AB}$ into convex combinations of rank-$1$ projectors, i.e.~with a probability distribution $\lset p_i\rset_{i=1}^m$ and pure states $\ket{\psi^{(i)}}_{AB}\in \C^{d_A}\otimes \C^{d_B}$, $1\leq i\leq m$. It is easy to see that $\text{SN}(\rho_{AB})\in \lset 1,\ldots, \min(d_A,d_B)\rset$ with  $\SN(\rho_{AB})=1$ iff $\rho_{AB}$ is separable. 

Note that the above definitions extend more generally to positive operators (i.e.~unnormalized quantum states). Sometimes it will be convenient to drop the normalization and we will do so for some of our results.

\subsection{Positive maps}

There is a powerful duality relation between the set of positive operators with Schmidt number below some value and certain cones of positive maps. A linear map $\mathcal{L}:\M_{d_A}\ra\M_{d_B}$ is called
\begin{itemize}
    \item \textbf{positive} if $\mathcal{L}(X)\in\M_{d_B}^+$ for any $X\in\M_{d_A}^+$.
    \item \textbf{$k$-positive} if $\ident_k\otimes \mathcal{L}:\M_{k}\otimes \M_{d_A}\ra \M_{k}\otimes\M_{d_B}$ is positive.
    \item \textbf{completely positive} if it is $k$-positive for any $k\in\N$.
\end{itemize}
Recall that a linear map $\mathcal{L}:\M_{d_A}\ra\M_{d_C}$ is completely positive iff its Choi matrix $C_\mathcal{L}\in \M_{d_C}\otimes \M_{d_B}$ given by 
\begin{equation*}
    C_\mathcal{L} := (\mathcal{L}_A\otimes \ident_B)(\me_{AB})
    \label{equ:ChoiIso}
\end{equation*}
is positive (see~\cite{CHOI}). Here $\me_{AB} = \proj{\ME}{\ME}_{AB}$, with $d_A=d_B$, is the maximally entangled state on $\C^{d_A}\otimes\C^{d_B}$, i.e.~the rank-$1$ projector on $\C^{d_A}\otimes\C^{d_B}$ corresponding to the pure state
\[
\ket{\ME}_{AB} = \frac{1}{\sqrt{d_A}}\sum^{d_A}_{i=1}\ket{i}_A\otimes \ket{i}_B\, .
\]
A paradigmatic example of a positive map that is not completely positive is given by the transposition map $\trans_d:\M_d\ra\M_d$, defined by $\trans_d(\ketbra{i}{j})=\ketbra{j}{i}$ in the computational basis (here $\proj{i}{j}\in\M_d$ denote the matrix units having a single $1$ in the $(i,j)$-entry).  Recall that the Choi matrix of the transposition map is given by $C_{\trans_d} = \flip/d$ where $\flip\in \M_d\otimes \M_d$ denotes the flip operator defined as $\flip\ket{ij} = \ket{ji}$ in the computational basis.

It turns out~\cite[Theorem 1]{Terhal00} that a quantum state $\rho_{AB}$ on $\C^{d_A}\otimes\C^{d_B}$ satisfies $\text{SN}(\rho_{AB})\leq n$ iff $(\ident_A\otimes \mathcal{P}_B)(\rho_{AB})\geq 0$ for any $n$-positive map $\mathcal{P}:\M_{d_B}\ra \M_{d_A}$. In particular it is separable iff $(\ident_A\otimes \mathcal{P}_B)(\rho_{AB})\geq 0$ for any positive map $\mathcal{P}:\M_{d_B}\ra \M_{d_A}$ (see~\cite{Horodecki96}). It is a hard problem in general to decide whether a given state is separable~\cite{Gurvits} and the set of positive maps has a very complicated structure. However, if the dimensions are low enough positive maps have a simple characterization~\cite{Woronowicz}: If $d_Ad_B\leq 6$ a linear map $\mathcal{L}:\M_{d_A}\ra \M_{d_B}$ is positive iff 
\begin{equation}
    \mathcal{L} = \mathcal{T} + \trans_{d_B}\circ \mathcal{T}'
\label{equ:decomposable}
\end{equation}
for completely positive maps $\mathcal{T},\mathcal{T}':\M_{d_A}\ra \M_{d_B}$. The structure of low dimensional positive maps implies that for $d_Ad_B\leq 6$ a quantum state $\rho_{AB}$ on $\C^{d_A}\otimes\C^{d_B}$ is separable iff its partial transpose is positive, i.e.~$\rho_{AB}^{\Gamma_B} \coloneqq (\ident_{A}\otimes \trans_B)(\rho_{AB})\geq 0$. In higher dimensions this statement turns out to be false, and in general there exist quantum states on $\C^{d_A}\otimes\C^{d_B}$ for $d_Ad_B> 6$ that are entangled despite having positive partial transpose (see \cite{Horodecki96} for some examples). In the following we will refer to quantum states having positive partial transpose simply as PPT states. 

The subset of positive maps $\mathcal{P}:\M_{d_A}\ra\M_{d_B}$ satisfying \eqref{equ:decomposable} for completely positive maps $\mathcal{T},\mathcal{T}':\M_{d_A}\ra \M_{d_B}$ are called decomposable. It turns out that a positive map $\mathcal{P}$ is decomposable iff $(\ident_A\otimes \mathcal{P}_B)(\rho_{AB})\geq 0$ for any PPT state $\rho_{AB}$ on $\C^{d_A}\otimes\C^{d_B}$ (see \cite{stormer}). Conversely, this means that the entanglement of PPT states can only be detected by positive maps that are non-decomposable.

\subsection{What is the maximal Schmidt number of a PPT state?}

Recall that for dimensions $d_A,d_B\in\N$ satisfying $d_Ad_B\leq 6$ the set of PPT states on $\C^{d_A}\otimes\C^{d_B}$ coincides with the set of separable states~\cite{Horodecki96}, i.e.~the states with Schmidt number $1$. Motivated by this result, it has been conjectured in~\cite{Sanpera01} that for $d_A=d_B=3$ the Schmidt number of any PPT state is less than $2$. Recently this conjecture has been proven by Yang et al.~\cite{Yang16}. Thus, it is a natural question to ask how large the Schmidt number $\SN(\rho_{AB})$ of a PPT state $\rho_{AB}$ on $\C^{d_A}\otimes\C^{d_B}$ can be, for fixed local dimensions $d_A,d_B > 3$. We will be particularly interested in the case where $d_A=d_B$. 

A partial answer to this question has been given in the asymptotic setting~\cite[\S 4.3]{SWZ} by comparing the volumes of the sets of PPT states and states of Schmidt number below a given value. Specifically, it has been shown that for sufficiently large $d=d_A=d_B$ there exist PPT states $\rho_{AB}$ on $\C^{d_A}\otimes\C^{d_B}$ with $\SN(\rho_{AB})\geq \alpha d$ for some constant $\alpha>0$ that is independent of the dimension $d$. It should be noted that the constant $\alpha$ in these results is hard to compute exactly, and probably quite small. Moreover, these methods do not yield explicit examples of PPT states exhibiting such Schmidt numbers.

Recently, an explicit construction of a PPT state $\rho_{AB}$ on $\C^{d_A}\otimes\C^{d_B}$ with $d=d_A=d_B$ and Schmidt number $\SN(\rho_{AB})\sim \log d$ has been given in~\cite[Proposition 20]{Chen17}. This is still far away from the linear dependence observed using volume estimates. A simple argument achieving a power law scaling $\SN(\rho_{AB}) \sim d^\gamma$ can be given using any convex, faithful and additive entanglement measure. For simplicity we will only state this argument for the squashed entanglement: Given a quantum state $\rho_{AB}$ on $\C^{d_A}\otimes\C^{d_B}$ its squashed entanglement is defined as~\cite{sq1,sq2}
\[
E_\text{sq}(\rho_{AB}) \coloneqq \inf\left\{ \frac{1}{2}I(A:B|C)_{\sigma}~:~d_C\in\N,\  \sigma_{ABC}\in(\M_{d_A}\otimes \M_{d_B}\otimes \M_{d_C})^+,~\mathrm{Tr}_C(\sigma_{ABC})=\rho_{AB}\right\}
\]
where $I(A:B|C)_\sigma \coloneqq S(AC)_\sigma + S(BC)_\sigma - S(ABC)_\sigma - S(C)_\sigma$ denotes the conditional mutual information of $\sigma_{ABC}$. It has been shown that the squashed entanglement satisfies
\begin{enumerate}
    \item $E_{\text{sq}}(\rho_{AB})>0$ iff $\rho_{AB}$ is entangled.
    \item $E_{\text{sq}}(\rho^{\otimes n}_{AB}) = n E_{\text{sq}}(\rho_{AB})$ for any $n\in\N$.
    \item $2^{E_{\text{sq}}(\rho_{AB})}\leq \text{SN}(\rho_{AB})$.
\end{enumerate}
See \cite{sq3} for a proof of the first, and \cite[Proposition 4]{sq2} for a proof of the second property. The third property is a simple consequence of the convexity of squashed entanglement (see \cite[Proposition 3]{sq2}). For any entangled PPT state $\rho_{AB}$ on $\C^{d_A}\otimes\C^{d_B}$ combining these properties of squashed entanglement yields
\[
\text{SN}(\rho^{\otimes n}_{AB}) \geq 2^{nE_\text{sq}(\rho_{AB})} \underset{n\ra\infty}{\longrightarrow} \infty.
\]
Therefore, setting $d_A=d_B\geq 3$ and defining 
\[
\gamma = \frac{E_\text{sq}(\rho_{AB})}{\log(d_A)}>0
\]
yields a PPT state $\sigma_{A^nB^n} = \rho^{\otimes n}_{AB}$ with local dimension $d=d^n_A$ and Schmidt number $\text{SN}(\sigma_{A^nB^n})\geq d^\gamma$. Note that $0<\gamma\leq 1$, and it is not known whether $\gamma= 1$ can be achieved using PPT states.  

In Section \ref{sec explicit} we will present a construction of PPT states $\rho_{AB}$ on $\C^{d_A}\otimes \C^{d_B}$ with $d=d_A=d_B$ and Schmidt numbers $\text{SN}(\rho_{AB})\geq \lceil(d-1)/4\rceil$, where $\ceil{x}\coloneqq \min\{ n\in \mathds{Z}: n\geq x\}$ is the ceiling function.. This is the first explicit construction of a family of PPT states achieving a Schmidt number that scales linearly in the local dimension. We leave it as an open problem to determine the best possible Schmidt number for PPT states, with any given local dimensions.

\section{Explicit constructions of PPT states with high Schmidt number} \label{sec explicit}

In this section we present an explicit construction of a family of PPT states whose Schmidt number scales linearly in the local dimension. We will start with a general construction and then discuss the properties of specific examples. It should be emphasized that our family contains examples answering a question recently posed by Chen et al.~\cite[Conjecture 36]{Chen17}: \emph{Is there a PPT state $\rho$ with an arbitrary large difference $\SN(\rho)-\SN(\rho^\Gamma)$?} We will discuss such an example in the second subsection.

\subsection{A family of states which are PPT and have high Schmidt number}

Let us start with a basic lemma introducing the general form of the states we want to consider.

\begin{lemma} \label{lemma positivity Z}
Consider an operator $Z\in \M_{d_1}\otimes \M_{d_1}\otimes \M_{d_2}\otimes \M_{d_2}$ of the form 
\begin{equation}
Z_{A_1B_1A_2B_2} = X_{A_1 B_1} \otimes (\id - \me)_{A_2 B_2} + Y_{A_1B_1} \otimes \me_{A_2 B_2}\, ,
\label{equ:DefZ}
\end{equation}
where $X,Y\in \M_{d_1}\otimes \M_{d_1}$ and where we label the first two tensor factors (of dimension $d_1$) as $A_1,B_1$ and the second two tensor factors (of dimension $d_2$) as $A_2,B_2$. Then $Z_{A_1B_1A_2B_2}$ is:
\begin{enumerate}[(a)]
\item positive semidefinite iff $X,Y\geq 0$;
\item PPT with respect to the bipartition $A_1A_2:B_1B_2$ iff $(d_2-1)X^\Gamma \geq - Y^\Gamma$ and $(d_2+1)X^\Gamma \geq Y^\Gamma$.
\end{enumerate}
\end{lemma}

\begin{proof}
Claim (a) is clear. To prove claim (b) we compute 
\begin{align*}
Z_{A_1B_1A_2B_2}^{\Gamma_{B_1B_2}} &= X_{A_1B_1}^{\Gamma_{B_1}} \otimes \lb\id - \frac{\flip}{d_2}\rb_{A_2 B_2} + Y_{A_1B_1}^{\Gamma_{B_1}} \otimes \lb\frac{\flip}{d_2}\rb_{A_2 B_2} \\
&= \left( X + \frac{Y - X}{d_2} \right)_{A_1B_1}^{\Gamma_{B_1}} \otimes \lb\frac{\id +\flip}{2}\rb_{A_2 B_2} + \left( X - \frac{Y - X}{d_2} \right)_{A_1B_1}^{\Gamma_{B_1}} \otimes \lb\frac{\id -\flip}{2}\rb_{A_2 B_2}\, ,
\end{align*}
which is positive iff the stated conditions hold.
\end{proof}

The previous lemma characterizes the cases where a matrix $Z_{A_1B_1A_2B_2}$ of the form \eqref{equ:DefZ} is positive and PPT (i.e.~an unnormalized PPT state). To characterize the entanglement of these states the following lemma will be useful.

\begin{lemma} \label{lemma detector}
Let $X, Y\in \M_{d_1}\otimes \M_{d_1}$ be two Hermitian operators such that
\begin{equation}
\exists\ \ket{\alpha}\in \C^{d_1}\otimes \C^{d_1}\hspace{0.15cm}:\hspace{0.15cm} \text{\emph{$\bra{\alpha}X\ket{\alpha} = 0$ and $\bra{\alpha}Y\ket{\alpha} >0$}} \,.
\label{detector eq1}
\end{equation}
Let $Z_{A_1B_1A_2B_2}\in \M_{d_1}\otimes \M_{d_1}\otimes \M_{d_2}\otimes \M_{d_2}$ denote the matrix from \eqref{equ:DefZ} with these $X,Y$. Then, for any $d_2'\in\N$, any linear map $\mathcal{L}:\M_{d_2}\ra \M_{d_2'}$ that is not completely positive satisfies
\begin{equation}
(\ident_{A_1}\otimes \ident_{B_1}\otimes \mathcal{L}_{A_2}\otimes \ident_{B_2})(Z_{A_1B_1A_2B_2})\ngeq 0\,.
\label{detector eq3}
\end{equation}
\end{lemma}

\begin{proof}
Note that 
\begin{equation*}
(\ident_{A_1}\otimes \ident_{B_1}\otimes \mathcal{L}_{A_2}\otimes \ident_{B_2})(Z_{A_1B_1A_2B_2}) = X_{A_1B_1} \otimes \lb \mathcal{L}(\id)\otimes \id - C_\mathcal{L}\rb_{A_2' B_2} + Y_{A_1 B_1} \otimes (C_{\mathcal{L}})_{A_2'B_2} \, .
\end{equation*}
Since $\mathcal{L}$ is not completely positive we have that the Choi matrix $C_\mathcal{L}$ is not positive semidefinite, i.e.~there exists $\ket{\beta}$ such that 
\begin{equation*}
 \bra{\beta}C_\mathcal{L}\ket{\beta}<0\,.  
\end{equation*}
With the vector $\ket{\alpha}$ from the assumptions of the lemma we get
\begin{equation*}
(\bra{\alpha}_{A_1B_1}\otimes \bra{\beta}_{A_2' B_2}) (\ident_{A_1}\otimes \ident_{B_1}\otimes \mathcal{L}_{A_2}\otimes \ident_{B_2})(Z_{A_1B_1A_2B_2}) (\ket{\alpha}_{A_1 B_1}\otimes \ket{\beta}_{A_2' B_2}) = \bra{\alpha} Y \ket{\alpha} \bra{\beta} C_\mathcal{L} \ket{\beta} < 0\, .
\end{equation*}
\end{proof}

With the previous lemmas we can now present the main result of this section.

\begin{thm}[Constructing PPT states with high Schmidt number]
Let $X,Y\in \M_{d_1}\otimes \M_{d_1}$ be such that 
\begin{enumerate}[(i)]
\item $X,Y\geq 0$;
\item $(d_2-1)X^\Gamma \geq - Y^\Gamma$ and $(d_2+1)X^\Gamma \geq Y^\Gamma$;
\item $\exists\ \ket{\alpha}\in \C^{d_1}\otimes \C^{d_1}\hspace{0.15cm}:\hspace{0.15cm} \text{\emph{$\bra{\alpha}X\ket{\alpha} = 0$ and $\bra{\alpha}Y\ket{\alpha} >0$.}}$
\end{enumerate}
Then the operator $Z_{AB} = Z_{A_1B_1A_2B_2}$ (with joint labels $A=A_1A_2$ and $B=B_1B_2$) defined in \eqref{equ:DefZ} is positive, PPT and satisfies
\begin{equation}
\SN(Z_{AB}) \geq  \ceil[\bigg]{\frac{d_2}{d_1}}\, .
\end{equation}
\label{thm:Family}
\end{thm}

\begin{proof}
By conditions (i) and (ii), applying Lemma~\ref{lemma positivity Z} shows that $Z_{AB}$ is positive and PPT. 

We can assume without loss of generality that $2d_1\leq d_2$ since otherwise the statement about the Schmidt number becomes trivial. Consider the Choi map $\mathcal{P}:\M_{d_2}\ra \M_{d_2}$ given by
\[
\mathcal{P} \coloneqq \one_{d_2}\mathrm{Tr} - \frac{1}{d_2-1}\ident_{d_2}\, .
\]  
It is well known (see~\cite{Tomiyama}) that $\mathcal{P}$ is $(d_2-1)$-positive, but not completely positive. Applying Lemma \ref{lemma detector} with $d_2'=d_2$ and $\mathcal{L}=\mathcal{P}$ shows that 
\[
(\ident_{A_1}\otimes \ident_{B_1}\otimes \mathcal{P}_{A_2}\otimes \ident_{B_2})(Z_{A_1B_1A_2B_2})\ngeq 0\, .
\]
Since the map $\ident_{A_1}\otimes \mathcal{P}_{A_2}$ is $\lfloor(d_2 -1)/d_1\rfloor$-positive we find that, with respect to the bipartition $A:B = A_1A_2:B_1B_2$,
\[
\SN(Z_{AB}) \geq  \floor[\bigg]{\frac{d_2 -1}{d_1}} + 1 = \ceil[\bigg]{\frac{d_2}{d_1}}\,.
\]
\end{proof}

We have not shown so far that there actually are operators $X,Y\in\M_{d_1}\otimes \M_{d_1}$ satisfying conditions (i),(ii) and (iii) in Theorem \ref{thm:Family}. A simple construction of such operators is shown in the following theorem.

\begin{thm}[Concrete family of PPT states with high Schmidt number]
Let $d_1,d_2\in \N$ with $d_1\leq d_2$ and let $X,Y\in \M_{d_1}\otimes \M_{d_1}$ be given by 
\[
X_{A_1B_1} = (\id - \me)_{A_1B_1}\quad \text{and }\quad Y_{A_1B_1} = (d_1-1)(d_2+1)\, \me_{A_1B_1}\, .
\]
Then the operator $Z_{AB} = Z_{A_1B_1A_2B_2}$ (with joint labels $A=A_1A_2$ and $B=B_1B_2$) defined in $\eqref{equ:DefZ}$ is positive, PPT and satisfies
\[
\SN(Z_{AB}) \geq  \ceil[\bigg]{\frac{d_2}{d_1}}\, .
\]
\label{thm:ConcreteEx}
\end{thm}

\begin{proof}
By a straightforward computation the operators $X,Y$ satisfy conditions (i),(ii) of Theorem \ref{thm:Family}, and condition (iii) follows by choosing $\ket{\alpha} = \ket{\ME}$.  
\end{proof}

The previous theorem provides a concrete family of PPT states with high Schmidt number. Note that these states appeared previously in \cite{Ishizaka04}, \cite{Piani07} and \cite{LLMH}. However, in these works their Schmidt number is not estimated. 

To further clarify the scaling of the Schmidt number with the local dimension we state the following simple corollary.

\begin{corollary} \label{corollary scaling}
For any $d\in \N$ with $d\geq 2$ there exists a PPT state $\rho$ on $\C^d\otimes\C^d$ satisfying 
\[ \SN\lb \rho\rb\geq \ceil[\bigg]{\frac{d-1}{4}}\,. \]
\end{corollary}

\begin{proof}
If $d =2d_2$ we can directly apply Theorem \ref{thm:ConcreteEx} with $d_1=2$. If $d=2d_2 + 1$ again apply Theorem \ref{thm:ConcreteEx} for dimensions $d_2$ and $d_1=2$ to obtain $\widetilde{Z}_{AB}\in \M_{2d_2}\otimes \M_{2d_2}$, which we can embed into the higher dimension by setting $Z_{AB} = \widetilde{Z}_{AB}\oplus 0$.
\end{proof}

The above Corollary~\ref{corollary scaling} constitutes an improvement over previous results, in particular over the explicit examples of~\cite{Chen17} ($\SN \sim \log d$) and over the implicit construction of~\cite[\S 4.3]{SWZ} ($\SN\sim \alpha d$ for some unknown universal constant $\alpha$). 

 To conclude this section we will briefly discuss some consequences of the above constructions to the theory of positive maps between matrix algebras. The following corollary has previously been obtained in \cite{Piani07}. We restate it here together with a slightly different proof, since it is not as well known as it should be.

\begin{cor}[Non-decomposability from tensoring with identity \cite{Piani07}]
Consider $k,d,d'\in\N$ with $k\geq 2$. For any linear map $\mathcal{L}:\M_{d}\ra\M_{d'}$ that is not completely positive, the linear map $\ident_k\otimes \mathcal{L}$ is not decomposable. 
\end{cor}

\begin{proof}
Note that for $k\geq d$ the map $\ident_k\otimes \mathcal{L}$ cannot be positive, and hence it is not decomposable. Therefore, we can assume that $k<d$. Consider the operator $Z_{AB} = Z_{A_1B_1A_2B_2}\in \M_{k}\otimes \M_k\otimes \M_{d}\otimes \M_d$ obtained from Theorem \ref{thm:ConcreteEx} (i.e.~setting $d_1=k$ and $d_2=d$). By construction $Z_{AB}$ is positive and PPT with respect to the bipartition $A:B=A_1A_2:B_1B_2$. Now applying Lemma \ref{lemma detector} for the linear map $\mathcal{L}$ shows that 
\[
(\ident_{A_1}\otimes \ident_{B_1}\otimes \mathcal{L}_{A_2}\otimes \ident_{B_2})(Z_{A_1B_1A_2B_2})\ngeq 0.
\]
By \cite{stormer} this shows that $\ident_{A_1}\otimes \mathcal{L}_{A_2} = \ident_k\otimes \mathcal{L}$ is not decomposable. 

\end{proof}

The previous corollary gives an easy method to construct non-decomposable positive maps. Starting from an $n$-positive map $\mathcal{P}:\M_{d_1}\ra\M_{d_2}$ that is not completely positive, the map $\ident_k\otimes \mathcal{P}$ for $k\leq n$ will be positive and not decomposable. Any state from the general family constructed in Theorem \ref{thm:Family} will serve as a witness (in the sense of \cite{stormer}) for this property. 

\subsection{Large variation of Schmidt number under partial transposition}

In this section we will answer the following question, recently raised in the literature~\cite[Conjecture 36]{Chen17}: \emph{Are there PPT states $\rho$ with an arbitrarily large difference $\SN(\rho)-\SN(\rho^\Gamma)$?} With the following theorem we demonstrate that this difference can be arbitrarily large as the local dimension grows.

\begin{theorem}
For any $d_1,d_2\in\N$ with $d_2\geq d_1$ the operator $Z_{AB} = Z_{A_1B_1A_2B_2}\in \M_{d_1}\otimes \M_{d_2}\otimes \M_{d_1}\otimes \M_{d_2}$ from Theorem \ref{thm:ConcreteEx} is positive, PPT and satisfies
\begin{equation}
    \SN\left(Z_{AB}^\Gamma\right) \leq 4\, ,
\end{equation}
and hence
\begin{equation}
    \SN(Z_{AB}) - \SN(Z_{AB}^\Gamma) \geq \ceil[\bigg]{\frac{d_2}{d_1}} - 4\, .
\end{equation}
\label{thm:SNdiff}
\end{theorem}

\begin{proof}
By Theorem \ref{thm:ConcreteEx} the operator $Z_{AB}$ is positive, PPT and satisfies
\[
\SN(Z_{AB}) \geq \ceil[\bigg]{\frac{d_2}{d_1}}.
\]
Taking the partial transpose with respect to the $B$ system we obtain
\begin{align*}
    &Z_{AB}^\Gamma = \id_{A_1B_1}\otimes\id_{A_2B_2} - \frac{1}{d_2}\,\id_{A_1B_1}\otimes\flip_{A_2B_2} - \frac{1}{d_1}\,\flip_{A_1B_1}\otimes\id_{A_2B_2} + \frac{d_1d_2+d_1-d_2}{d_1d_2}\,\flip_{A_1B_1}\otimes \flip_{A_2B_2}\, .
\end{align*}
The four summands appearing in the above expression are mutually commuting operators, and therefore it is easy to diagonalise $Z_{AB}^\Gamma$. The eigenvectors of $Z_{AB}^\Gamma$ are of the form
\begin{equation*}
    \ket{\Psi_{jklm}} \coloneqq \ket{\chi_{jk}}_{A_1B_1}\otimes \ket{\chi_{lm}}_{A_2B_2}\, ,
\end{equation*}
with $1\leq j,k\leq d_1$, $1\leq l,m\leq d_2$, and where 
\begin{equation*}
    \ket{\chi_{pq}} \coloneqq \left\{ \begin{array}{lll} \frac{1}{\sqrt2}(\ket{pq}+\ket{qp}) & & \text{if $p<q$,} \\[0.5ex] \ket{pp} & & \text{if $p=q$,} \\[0.5ex]
    \frac{1}{\sqrt2} (\ket{pq}-\ket{qp}) & & \text{if $p>q$.} \end{array}\right.
\end{equation*}
It is easy to see that the Schmidt rank of each vector $\ket{\Psi_{jklm}}$ is upper bounded by $4$ with respect to the bipartition $A:B = A_1A_2:B_1B_2$. Since $Z_{AB}^\Gamma$ is positive semidefinite this shows that 
\[
\SN(Z_{AB}^\Gamma)\leq 4\,.
\]
\end{proof}

We obtain the following immediate corollary.

\begin{cor}
For any $d\in \N$ with $d\geq 2$ there exist a PPT entangled state $\rho$ on $\C^d\otimes\C^d$ satisfying
\[ \SN(\rho)-\SN(\rho^\Gamma)\geq \ceil[\bigg]{\frac{d-1}{4}} - 4 \,. \]
\end{cor}

\begin{proof}
The statement follows by combining Theorem \ref{thm:SNdiff} with Corollary~\ref{corollary scaling}.
\end{proof}

\subsection{States which are invariant under partial transpose and have high Schmidt number}

The results obtained in the previous subsection show that the Schmidt numbers of the states constructed in Theorem \ref{thm:Family} can change dramatically by applying a partial transposition. It is therefore a natural question to ask what values the Schmidt number can take on states that are invariant under partial transposition. Note for a bipartite state there are two possibilities to define invariance under partial transposition depending on which subsystem we choose the transposition to act on. Since we will aim at states on $\C^d\otimes \C^d$ (i.e. where the two subsystems have the same dimension) we will use the term PT-invariant and it will be clear from context which of the two possibilities we choose.

In Section \ref{sec:UpperBounds} we show that PT-invariant states, where the transposition is taken on the smaller subsystem, the Schmidt number cannot reach the highest possible value. In this section, we show that nevertheless there do exist states that are invariant under partial transposition and have a Schmidt number which scales linearly with the local dimension.

In the following we denote by $\ket{+i}, \ket{-i}\in \C^2$ the pure states 
\[
\ket{+i} = \frac{1}{\sqrt{2}}(\ket{0}+i\ket{1})\text{ and } \ket{-i} = \frac{1}{\sqrt{2}}(\ket{0}-i\ket{1})
\]
where $\{\ket{0},\ket{1}\}$ denotes the computational basis of $\C^2$ (the standard basis for our transposition). Note that $\ket{+i}$ and $\ket{-i}$ are orthonormal. The following lemma is probably well known. 

\begin{lem}
Given a PPT state $\rho_{AB}$ on $\C^{d_A}\otimes\C^{d_B}$, the state $\tilde{\rho}_{ABB'}$ on $\C^{d_A}\otimes\C^{d_B}\otimes\C^2$ given by 
\[
\tilde{\rho}_{ABB'} = \rho_{AB}\otimes \proj{+i}{+i}_{B'} + \rho^{\Gamma_B}_{AB}\otimes \proj{-i}{-i}_{B'}
\]
is invariant under partial transposition on $BB'$.
\label{lem:MakePTInv}
\end{lem}

\begin{proof}
Note that 
\[
\trans(\proj{+i}{+i}) = \proj{-i}{-i}\,.
\]
And it thus easily follows that
\[
\tilde{\rho}_{ABB'}^{\Gamma_{BB'}} = \rho_{AB}^{\Gamma_B}\otimes \trans(\proj{+i}{+i}_{B'}) + \rho_{AB}\otimes \trans(\proj{-i}{-i}_{B'}) = \tilde{\rho}_{ABB'}\,.
\]

\end{proof}

With the previous lemma we can easily modify the family of states constructed in Theorem \ref{thm:Family} to make it invariant under partial transposition. To keep the discussion simple we will only state the result on the maximal Schmidt number that can be obtained using this construction. 

\begin{thm}
For any $d\geq 4$ there exists a PT-invariant state $\rho$ on $\C^d\otimes\C^d$ satisfying 
\[ \SN\lb \rho\rb\geq \begin{cases}\ceil[\big]{\frac{d-2}{8}},& \text{ if }d\text{ is even,}\vspace{0.1cm}\\\ceil[\big]{\frac{d-3}{8}},&\text{ else.} \end{cases}\, \]
\end{thm}

\begin{proof}
Consider first the case where $d$ is even. For $d'=d/2$ consider the PPT state $\sigma_{AB}$ on $\C^{d'}\otimes\C^{d'}$ constructed in Corollary \ref{corollary scaling} with 
\[
\SN\lb \sigma_{AB}\rb\geq \ceil[\bigg]{\frac{d'-1}{4}}\,.
\]
By Lemma \ref{lem:MakePTInv} the state
\[
\tilde{\sigma}_{ABB'} = \sigma_{AB}\otimes \proj{+i}{+i}_{B'} + \sigma^{\Gamma_B}_{AB}\otimes \proj{-i}{-i}_{B'},
\]
on $\C^{d'}\otimes\C^{d'}\otimes\C^2$ is invariant under partial transposition on $BB'$. Since the Schmidt number is decreasing under the application of local completely positive maps (see \cite{Terhal00}), and 
\[
\sigma_{AB} = (\one_{AB}\otimes \bra{+i}_{B'})\tilde{\sigma}_{ABB'}(\one_{AB}\otimes \ket{+i}_{B'})\, ,
\]
we have that, with respect to the bipartition $A:BB'$, 
\[
\SN\lb \tilde{\sigma}_{ABB'}\rb\geq \ceil[\bigg]{\frac{d'-1}{4}}\, .
\]
Now in order to make the two local dimensions equal we just consider the state
\[
\rho_{AA'BB'} = \tilde{\sigma}_{ABB'}\otimes \proj{0}{0}_{A'}\,
\]
on $\C^{d'}\otimes\C^2\otimes\C^{d'}\otimes\C^2$. With respect to the bipartition $AA':BB'$, we have
\[
\SN\lb \rho_{AA'BB'}\rb\geq \ceil[\bigg]{\frac{d'-1}{4}}\, .
\]
Since $d=2d'$ the statement of the theorem follows. 

In the case where $d$ is odd we can consider the state $\rho_{AA'BB'}$ on $\C^{d-1}\otimes \C^{d-1}$ obtained from the above construction. Then the embedding $\rho_{AA'BB'}\oplus 0$ of this state into $\M_d\otimes \M_d$ (extending each of the local dimensions by $1$) satisfies the Schmidt number bound stated in the theorem.

\end{proof}

\section{Typical Schmidt number of PPT states in high dimensions}
\label{sec:HighDim}

Having constructed examples of PPT states requiring high dimensions it is now natural to ask: \emph{Is this phenomenon generic?} In other words we want to find out whether PPT states with high Schmidt numbers are just an oddity that may physically not be relevant. We will approach this question from two angles: First, we compare the relative volumes occupied by PPT states and states with bounded Schmidt number respectively in the set of all states. Then, we derive probabilities for random states to be both PPT and of high Schmidt number. 

To start this section, let us fix some further notation. For any $1\leq k\leq d$, we denote by $\mathrm{SN}_k(\C^d\,{:}\,\C^d)$ the set of states on $\C^d\otimes\C^d$ which have Schmidt number at most $k$. In particular, $\mathrm{SN}_1(\C^d\,{:}\,\C^d)$ is just the set of separable states and $\mathrm{SN}_d(\C^d\,{:}\,\C^d)$ is the set of all states on $\C^d\otimes\C^d$.
We also denote by $\PPTSet(\C^d\,{:}\,\C^d)$ the set of states on $\C^d\otimes\C^d$ which are PPT.

\subsection{Sizes of the sets of bounded Schmidt number states vs PPT states}

To begin with, we would like to know how the sizes of the sets $\mathrm{SN}_k(\C^d\,{:}\,\C^d)$ and $\PPTSet(\C^d\,{:}\,\C^d)$ compare to one another. Here, we are interested in the high-dimensional regime, i.e.~when $d$ is large. There are two size parameters that we will be looking at: the \emph{mean width} and the \emph{volume radius}. Informally, given a convex body $\mathrm{K}$, its mean width is defined as the distance between an origin point and the tangent hyperplane to $\mathrm{K}$ in some direction, averaged over all directions, while its volume radius is defined as the radius of the Euclidean ball which would have the same volume as $\mathrm{K}$. The precise mathematical definitions appear below.

The mean width is defined as follows: Given $\mathrm{K}(\C^n)$ a convex set of states on $\C^n$,
\[ w\big( \mathrm{K}(\C^n) \big) \coloneqq \E\, \sup \left\{ \tr(X\sigma) :\ \sigma\in \mathrm{K}(\C^n) \right\} \, , \]
for $X$ uniformly distributed on the Hilbert--Schmidt unit sphere of the set of trace-$0$ Hermitian operators on $\C^n$. Equivalently, we can rewrite
\begin{equation} \label{eq:defw} w\big( \mathrm{K}(\C^n) \big) = \frac{1}{\E\|G\|_2} \,\E \,\sup \left\{ \tr(G\sigma):\ \sigma\in\mathrm{K}(\C^n) \right\} \underset{n\rightarrow\infty}{\sim} \frac{1}{n}\,\E \, \sup \left\{ \tr(G\sigma):\ \sigma\in\mathrm{K}(\C^n) \right\}\, , \end{equation}
for $G$ a trace-$0$ matrix from the Gaussian Unitary Ensemble (GUE) on $\C^n$ (see e.g.~\cite[Chapter 2]{AGZ} for a proof of the last asymptotic estimate). Since we will be manipulating repeatedly such random matrix in the remainder of this section, let us recall here its precise definition: Start from $\tilde{G}$ an $n\times n$ matrix whose entries are independent complex Gaussian variables (with mean $0$ and variance $1$). Then, $G'=(\tilde{G}+\tilde{G}^{\dagger})/\sqrt{2}$ is an $n\times n$ GUE matrix and $G=G'-(\tr G')\openone/n$ is a trace-$0$ $n\times n$ GUE matrix. In other words, $G'$, resp.~$G$, is simply the standard Gaussian vector in the space of Hermitian operators, resp.~trace-$0$ Hermitian operators, on $\C^n$.

The volume radius is defined as follows: Given $\mathrm{K}(\C^n)$ a convex set of states on $\C^n$,
\[ \mathrm{vrad}\big( \mathrm{K}(\C^n) \big) \coloneqq \left( \frac{\mathrm{Vol}\big( \mathrm{K}(\C^n) \big)}{\mathrm{Vol}\big( \mathrm{B}^{n^2-1} \big)} \right)^{1/(n^2-1)}\, , \]
where $\mathrm{B}^{n^2-1}$ stands for the Hilbert--Schmidt unit ball of the set of trace-$1$ Hermitian operators on $\C^n$ (which can be identified with the real Euclidean unit ball of dimension $n^2-1$) and $\mathrm{Vol}(\cdot)$ for the $(n^2-1)$-dimensional Lebesgue measure.

By Urysohn's inequality (see e.g.~\cite[Corollary 1.4]{Pisier}), we know that
\[ \mathrm{vrad}\big(\mathrm{K}(\C^n)\big) \leq w\big(\mathrm{K}(\C^n)\big)\, ,\]
and in many cases, these two quantities are actually of the same order.

The mean width and the volume radius of the sets of bounded Schmidt number states and PPT states were estimated in~\cite{SWZ} and~\cite{AS}, respectively. We will restate these results for completeness:

\begin{theorem}[Size of the set of bounded Schmidt number states {\cite[Section 4]{SWZ}}] \label{th:w-vrad(S_k)}
There exist universal constants $c,C>0$ such that, for any $d\in\N$ and $1\leq k\leq d$, we have
\[ c\,\frac{\sqrt{k}}{d\sqrt{d}} \leq \mathrm{vrad}\big( \mathrm{SN}_k(\C^d\,{:}\,\C^d) \big) \leq w\big( \mathrm{SN}_k(\C^d\,{:}\,\C^d) \big) \leq C\,\frac{\sqrt{k}}{d\sqrt{d}}\, . \]
\end{theorem}

\begin{theorem}[Size of the set of PPT states {\cite[Theorem 4]{AS} and~\cite[Theorem 9.13]{ASbook}}] \label{th:w-vrad(P)}
For any $d\in\N$, we have
\[ \frac{1}{4d} \leq \mathrm{vrad}\big( \PPTSet(\C^d\,{:}\,\C^d) \big) \leq w\big( \PPTSet(\C^d\,{:}\,\C^d) \big) \leq \frac{2}{d}\, . \]
\end{theorem}

 Theorems~\ref{th:w-vrad(S_k)} and~\ref{th:w-vrad(P)} show that on $\C^d\otimes\C^d$ and for $k\ll d$ the volume of the set of PPT states is much bigger than the volume of the set of states which have Schmidt number at most $k$. Moreover, choosing $k=\alpha d$ for some $\alpha$ small enough (depending on the unknown constants $c$ and $C$) the theorems show that for $d$ large enough most PPT states on $\C^d\otimes\C^d$ have Schmidt numbers higher than $\alpha d$.

\subsection{Random construction of PPT states with high Schmidt number}

Let us now explain how one could exhibit bipartite quantum states which have both properties of being PPT and having a high Schmidt number. These states will be constructed at random, in such a way that one can argue that, with high probability, they meet these two conditions simultaneously. Our random state model is basically the same as the one considered in~\cite[Section V]{GHLS}, which we recall here. Let $G$ be a trace-$0$ GUE matrix on $\C^d\otimes\C^d$, and define the associated `maximally mixed + Gaussian noise' state on $\C^d\otimes\C^d$ as
\begin{equation} \label{eq:random-state} \rho := \frac{1}{d^2}\left(\id + \frac{\alpha }{d}G \right)\, , \end{equation}
where $0<\alpha<1/2$ is a fixed parameter. 

\begin{theorem} \label{prop:PPT}
Let $\rho$ be a random state on $\C^d\otimes\C^d$, as defined by equation~\eqref{eq:random-state}. Then,
\[ \P\left( \rho\in\PPTSet(\C^d\,{:}\,\C^d) \right) \geq 1-2e^{-c_{\alpha}d^2}\, , \]
where $c_{\alpha}>0$ is a constant depending only on the parameter $\alpha$.
\end{theorem}

\begin{proof}
The argument is exactly the same as in the proof of~\cite[Proposition V.5]{GHLS}. We briefly repeat it here for the sake of completeness.
$G$ and $G^{\Gamma}$ are both trace-$0$ GUE matrices on $\C^d\otimes\C^d$. Hence, we know from~\cite{Aubrun}, that for $H$ being either $G$ or $G^{\Gamma}$, we have
\[ \forall\ \varepsilon>0,\ \P(\lambda_{\min}(H) < -(2+\varepsilon)) \leq e^{-c\varepsilon^{3/2}d^2}\, . \]
Applying this deviation probability estimate to $\varepsilon=1/\alpha-2>0$, we get by the union bound that
\[ \P(\rho\geq 0\ \text{and}\ \rho^{\Gamma}\geq 0 ) \geq 1- 2e^{-c(1/\alpha-2)^{3/2}d^2}\, , \]
which is precisely the advertised result.
\end{proof}

\begin{theorem} \label{prop:SR}
	Let $\rho$ be a random state on $\C^d\otimes\C^d$, as defined by equation~\eqref{eq:random-state}. Then, for any $1\leq k\leq c'_{\alpha}d$, we have
	\[ \P\left( \rho\notin\mathrm{SN}_k(\C^d\,{:}\,\C^d) \right) \geq 1-2e^{-cd}\, , \]
	where $c'_{\alpha}>0$ is a constant depending only on the parameter $\alpha$ and $c>0$ is a universal constant.
\end{theorem}

\begin{proof}
The argument follows the exact same lines as in the proof of~\cite[Proposition V.6]{GHLS}. We might thus skip a few details here. The strategy is to exhibit a Hermitian operator $M$ on $\C^d\otimes\C^d$ which is with probability greater than $1-2e^{-cd}$ a witness of the fact that $\rho$ has Schmidt number larger than $k$. 

To begin with observe that
\[ \E\tr(\rho G) = \frac{\alpha}{d^3}\E\tr(G^2) = \alpha d\, , \]
while Theorem~\ref{th:w-vrad(S_k)} gives us
\[ \E\sup \left\{ \tr(\sigma G):\ \sigma\in\mathrm{SN}_k(\C^d\,{:}\,\C^d) \right\} \leq C\sqrt{k}\sqrt{d}\, . \]
With the two average estimates above at our disposal, we are now in position to apply the following Gaussian deviation inequality: If $f$ is a function such that, for any Gaussian variables $H,H'$, $|f(H)-f(H')|\leq L(H,H')\|H-H'\|_2$, with $\E L(H,H')\leq L$, then for any $\varepsilon>0$, $\P(|f-\E f|>\varepsilon)\leq e^{-c'\epsilon^2/L^2}$ (see e.g.~\cite[Chapter 2]{Pisier} for a proof). The Lipschitz constants of the two functions we are interested in were upper bounded in the proof of~\cite[Proposition V.6]{GHLS}. Using these estimates we thus get that, for any $\varepsilon>0$, on the one hand
\[ 
\P\big( \tr(\rho G) < (1-\varepsilon)\alpha d \big) \leq e^{-c\varepsilon^2d^4}\, , 
\]
and on the other hand
\[
\P\left( \sup\left\{ \tr(\sigma G):\ \sigma\in\mathrm{SN}_k(\C^d\,{:}\,\C^d) \right\} > (1+\varepsilon)C\sqrt{k}\sqrt{d} \right) \leq e^{-c\varepsilon^2d}\, . 
\]
Therefore, taking $\varepsilon=1/2$ in the two deviation probability estimates above, we get by the union bound that $M=\id-2G/\alpha d$ is such that
\[ \P \left( \tr(\rho M) < 0\ \text{and}\ \sup\left\{ \tr(\sigma M):\ \sigma\in\mathrm{SN}_k(\C^d\,{:}\,\C^d) \right\} > 1-\frac{3C}{\alpha}\frac{\sqrt{k}}{\sqrt{d}} \right) \geq 1- e^{-cd^4/4} - e^{-cd/4}\, . \]
The statement of the theorem follows from observing that $(3C/\alpha)\sqrt{k}/\sqrt{d}<1$ for $k<(\alpha^2/9C^2)d$ (and from relabelling $c/4$ in $c$).
\end{proof}

\begin{corollary} \label{cor:PPT-SR}
	Let $\rho$ be a random state on $\C^d\otimes\C^d$, as defined by equation~\eqref{eq:random-state}, with $1/8\leq\alpha\leq 3/8$. Then, for any $1\leq k\leq c'd$, we have
	\[ \P\left( \rho\in\PPTSet(\C^d\,{:}\,\C^d)\ \text{and}\ \rho\notin\mathrm{SN}_k(\C^d\,{:}\,\C^d) \right) \geq 1-4e^{-\hat{c}d}\, , \]
	where $\hat{c},c'>0$ are universal constants.
\end{corollary}

\begin{proof}
Corollary~\ref{cor:PPT-SR} is a direct consequence of Propositions~\ref{prop:PPT} and~\ref{prop:SR}. Indeed, since the constant $c'_{\alpha}$ of Proposition~\ref{prop:SR} is increasing with $\alpha$, one can choose $c'=c'_{1/8}$, and since the constant $c_{\alpha}$ of Proposition~\ref{prop:PPT} is decreasing with $\alpha$, one can choose $\hat{c}=\min(c_{3/8},c)$ (with the constant $c$ of Proposition~\ref{prop:SR}). 
\end{proof}

An explicit value for the universal constant $c'>0$ appearing in Corollary~\ref{cor:PPT-SR} could be estimated. It is nevertheless expected to be smaller than the constant $1/4$ obtained in our explicit construction. The main interest of Corollary~\ref{cor:PPT-SR} thus resides more in showing that, on large bipartite quantum systems, it is actually a typical feature of PPT states to have a Schmidt number which scales linearly with the local dimension.

\section{Upper bounds on the Schmidt number}
\label{sec:UpperBounds}
The previous sections show that quantum states can have high Schmidt number although they stay positive under partial transposition. We will now focus on the converse question: \emph{Can we infer limitations on the Schmidt number from positive partial transpose?} Our main results show such limitations for states invariant under partial transposition and for absolutely PPT states (precise definitions will be given later). These results are all based on a general technique relating the Schmidt number to entangled principal sub-blocks of a quantum state.

To state our results, we will need the notion of a trivial lifting of a linear map first considered in~\cite{Yang16}.

\begin{defn}[Trivial lifting~\cite{Yang16}]
 A linear map $\mathcal{L}:\M_{d_1}\ra \M_{d_2}$ is called an \emph{$\mathcal{S}$-trivial lifting} for a set $\mathcal{S}\subseteq \lset 1,\ldots ,d_1\rset$ iff $\mathcal{L}\lb\proj{i}{j}\rb=0$ whenever $i\in \mathcal{S}$ or $j\in \mathcal{S}$. 
\end{defn}

Note that in the above definition we do not care about the map that gets lifted to a particular $\mathcal{S}$-trivial lifting $\mathcal{L}:\M_{d_1}\ra \M_{d_2}$. This omission simplifies our presentation slightly, and we refer to \cite{Yang16} for a more elaborate presentation.  A remarkable decomposition technique for $k$-positive maps (called `Choi decomposition') has been proved in~\cite{Yang16}. We will need the following corollary of this result:

\begin{thm}[\cite{Yang16}]
For $k\in \lset 2,\ldots, \min(d_1,d_2)\rset$ any $k$-positive map $\mathcal{P}:\M_{d_1}\ra \M_{d_2}$ can be written as 
\[
\mathcal{P} = \mathcal{Q} + \mathcal{T}\, ,
\]
where $\mathcal{T}:\M_{d_1}\ra \M_{d_2}$ is completely positive and $\mathcal{Q}:\M_{d_1}\ra \M_{d_2}$ is a positive $\mathcal{S}$-trivial lifting for some $\mathcal{S}\subseteq\lset 1,\ldots ,d_1\rset$ with $\left|\mathcal{S}\right| = k-1$.
\label{thm:Decomp}
\end{thm}

The main consequence of the previous theorem is a dimension reduction for $k$-positive maps $\mathcal{P}:\M_{d_1}\ra \M_{d_2}$. Indeed such maps exhibit a behaviour that is not completely positive only on a subspace of their input space $\M_{d_1}$. We will relate this in the following to the Schmidt number.

\subsection{Schmidt numbers imply entanglement in principal sub-blocks}

The following theorem is the main technical tool of this section.

\begin{thm}[Entangled sub-blocks from Schmidt number]
For $d_1\leq d_2$ consider a matrix $X\in (\M_{d_1}\otimes \M_{d_2})^+$ written as  
\[
X = \sum^{d_1}_{i,j=1} \proj{i}{j}\otimes X_{ij}\, ,
\]
with blocks $X_{ij}\in\M_{d_2}$ for $i,j\in\lset 1,\ldots ,d_1\rset$. If $2\leq \SN\lb X\rb = k$, then there exists a set \[
\lset m_1,\ldots , m_{d_1-k+2}\rset\subseteq \lset 1,\ldots , d_1\rset\]
such that the principal sub-block matrix
\[
Y = \sum^{d_1 -k+2}_{s,t=1} \proj{s}{t}\otimes X_{m_sm_t}\in (\M_{d_1-k+2}\otimes \M_{d_2})^+
\]
is entangled.

\end{thm}
\begin{proof}
Since $\SN\lb X\rb = k$ there exists a $(k-1)$-positive map $\mathcal{P}:\M_{d_1}\ra \M_{d_2}$ such that 
\begin{equation}
(\mathcal{P}\otimes \mathrm{id}_{d_2})\lb X\rb\ngeq 0\, .
\label{equ:wit}
\end{equation}
By Theorem~\ref{thm:Decomp} there exists a set $\mathcal{S}\subset \lset 1,\ldots ,d_1\rset$ with $\left|\mathcal{S}\right| = k-2$ such that
\[
\mathcal{P} = \mathcal{Q} + \mathcal{T}\, ,
\]
where $\mathcal{T}:\M_{d_1}\ra \M_{d_2}$ is completely positive and $\mathcal{Q}:\M_{d_1}\ra \M_{d_2}$ is a positive $\mathcal{S}$-trivial lifting. Setting 
\[
\lset m_1,\ldots , m_{d_1-k+2}\rset = \lset 1,\ldots , d_1\rset\setminus \mathcal{S}\, ,
\]
we can conclude from \eqref{equ:wit} that 
\[
(\mathcal{Q}\otimes \mathrm{id}_{d_2})\lb X\rb = \sum^{d_1}_{i,j=1} \mathcal{Q}(\proj{i}{j})\otimes X_{ij} = \sum^{d_1-k+2}_{s,t=1} \mathcal{Q}(\proj{m_s}{m_t})\otimes X_{m_sm_t} \ngeq 0\, , 
\]
where we used that $\mathcal{Q}(\proj{i}{j})=0$ whenever $i\in \mathcal{S}$ or $j\in \mathcal{S}$. Since the map $\mathcal{Q}$ is positive we have that the positive matrix
\[
\sum^{d_1-k+2}_{s,t=1} \proj{m_s}{m_t}\otimes X_{m_tm_s} = (V\otimes \one_{d_2})Y (V^\dagger\otimes \one_{d_2})
\]
is entangled. Here $V:\C^{d_1-k+2}\ra \C^{d_1}$ is the isometry defined by $V\ket{s} = \ket{m_s}$ for $s\in \lset 1,\ldots ,d_1-k+2\rset$, and since the application of a local isometry preserves separability the proof is finished.

\end{proof}

For convenience we restate as a corollary the special case where the previous theorem is applied to a positive matrix with maximal Schmidt number.

\begin{cor}%[Maximal Schmidt number]
For $d_1\leq d_2$ consider a matrix $X\in (\M_{d_1}\otimes \M_{d_2})^+$ written as  
\[
X = \sum^{d_1}_{i,j=1} \proj{i}{j}\otimes X_{ij}\, ,
\]
with blocks $X_{ij}\in\M_{d_2}$ for $i,j\in\lset 1,\ldots ,d_1\rset$. If $\SN\lb X\rb = d_1$, then there exists $\lset k_1,k_2\rset\subset \lset 1,\ldots ,d_1\rset$ such that the principal sub-block matrix
\begin{equation}
Y = \sum^2_{s,t=1} \proj{s}{t}\otimes X_{k_sk_t}\in (\M_{2}\otimes \M_{d_2})^+
\label{equ:subblock}
\end{equation}
is entangled.

\label{cor:RedMaxSch}
\end{cor}

It is a simple consequence of the previous corollary that a state $\rho$ on $\C^{d_1}\otimes \C^{d_2}$ with $d_1\leq d_2$ has Schmidt number $\text{SN}(\rho)\leq d_1-1$ (i.e.~the Schmidt number is not maximal) if none of the principal sub-block matrices from \eqref{equ:subblock} is entangled. We will use this technique in the following to show that the Schmidt number of certain subsets of PPT states is not maximal.

\subsection{States which are invariant under partial transpose}

The following theorem shows that a state which is PT-invariant with respect to the smaller of the two subsystems cannot have maximal Schmidt number. In the case where the smaller system has dimension $2$ this had been shown in~\cite{2xNseparability}.   

\begin{thm}[Schmidt number of PT-invariant states]
If a quantum state $\rho$ on $\C^{d_1}\otimes\C^{d_2}$ with $2\leq d_1\leq d_2$ satisfies $(\trans_{d_1}\otimes \mathrm{id}_{d_2})(\rho) = \rho$, then $\SN\lb \rho\rb\leq d_1 -1$. 

\end{thm}
\begin{proof}
We write $\rho\in(\M_{d_1}\otimes \M_{d_2})^+$ as 
\[
\rho = \sum^{d_1}_{i,j=1} \proj{i}{j}\otimes X_{ij}
\]
with matrices $X_{ij}\in\M_{d_2}$ for $i,j\in\lset 1,\ldots ,d_1\rset$. Since $\rho$ is invariant under a partial transpose we have that
\[
\rho = \sum^{d_1}_{i,j=1} \proj{i}{j}\otimes X_{ij} = \sum^{d_1}_{i,j=1} \proj{i}{j}\otimes X_{ji} = (\trans_{d_1}\otimes \mathrm{id}_{d_2})(\rho),
\]
and thus $X_{ij} = X_{ji}$ for any $i,j\in\lset 1,\ldots ,d_1 \rset$. For any $k_1,k_2\in\lset 1,\ldots ,d_1\rset$ the matrix 
\[
\rho^{k_1 k_2} = \sum^2_{s,t=1} \proj{s}{t}\otimes X_{k_sk_t}
\] 
is positive as a principal sub-block matrix of the positive matrix $\rho$. 

Clearly, we have  
\[
(\trans_{d_1}\otimes \mathrm{id}_{d_2})(\rho^{k_1 k_2}) = \sum^2_{s,t=1} \proj{k_s}{k_t}\otimes X_{k_tk_s} = \sum^2_{s,t=1} \proj{k_s}{k_t}\otimes X_{k_sk_t} = \rho^{k_1 k_2}\, .
\]
Now \cite[Theorem 2]{2xNseparability} implies that $\rho^{k_1 k_2}\in(\M_2\otimes \M_{d_2})^+$ has to be separable (for a simple proof of this fact, see~\cite[\S 6]{revisited}). Therefore, $\SN(\rho) \leq d_1-1$ since otherwise Corollary~\ref{cor:RedMaxSch} would provide $k_1,k_2\in\lset 1,\ldots ,d_1\rset$ such that $\rho^{k_1 k_2}$ is entangled.  
\end{proof}

\subsection{States which are absolutely positive under partial transpose}

We now give an application of our techniques to the problem of classifying the absolutely PPT set. Let $\mathcal{U}(\C^d)\subset \M_d$ denote the set of unitary matrices acting on $\C^d$. A quantum state $\rho$ on $\C^{d_1}\otimes\C^{d_2}$ is called 
\begin{itemize}
    \item \emph{absolutely PPT} (APPT) if $U\rho U^\dagger$ is PPT for any unitary $U\in \mathcal{U}(\C^{d_1}\otimes\C^{d_2})$.
    \item \emph{absolutely separable} (ASEP) if $U\rho U^\dagger$ is separable for any unitary $U\in \mathcal{U}(\C^{d_1}\otimes\C^{d_2})$.
\end{itemize}
It has been shown in \cite{Howard} that there exists a ball of ASEP states around the maximally mixed state. Despite this not much is known about the structure of the set of ASEP states. In contrast the set of APPT states can be characterized in terms of a semidefinite program \cite{JohnstonSDP}. It is an open problem whether the sets of APPT states and ASEP states coincide for all dimensions $d_1$ and $d_2$ \cite{JohnstonSDP}. Here we prove that, at least, APPT states cannot have maximal Schmidt number. 

\begin{thm}[Schmidt number of APPT states]
Any APPT state $\rho$ on $\C^{d_1}\otimes\C^{d_2}$ satisfies
\[
\mathrm{SN}(\rho)\leq \min(d_1,d_2)-1\, .
\]

\end{thm}
\begin{proof}
Without loss of generality we assume that $d_1\leq d_2$. We can write 
\[
\rho = \sum^{d_1}_{i,j=1}\proj{i}{j}\otimes X_{ij},
\]
with matrices $X_{ij}\in\M_{d_2}$ for $i,j\in\lset 1,\ldots ,d_1\rset$. For any $k_1, k_2\in \lset 1,\ldots , d_1\rset$ with $k_1\neq k_2$ the matrix
\[
\rho^{k_1k_2} = \sum^2_{s,t=1} \proj{s}{t}\otimes X_{k_sk_t} 
\]
is positive as a principal sub-block matrix of the positive matrix $\rho$. We will show that $\rho^{k_1k_2}$ is separable for any $k_1,k_2\in \lset 1,\ldots , d_1\rset$. Once this is done, the statement of our theorem follows from Corollary \ref{cor:RedMaxSch}. 

To show that $\rho^{k_1k_2}$ for fixed $k_1,k_2\in \lset 1,\ldots , d_1\rset$ is separable consider a unitary $\tilde{U}\in\mathcal{U}(\C^2\otimes \C^{d_2})$. We can write
\[
\tilde{U} = \sum^2_{s,t=1}\proj{s}{t}\otimes V_{st}\, ,
\]
where $V_{st}\in \M_{d_2}$ for $s,t\in\lset 1,2\rset$. Now consider
\[
U = \sum^2_{s,t=1} \proj{k_s}{k_t}\otimes V_{s t} + \sum_{ i\in \lset 1,\ldots ,d_1\rset\setminus\lset k_1 ,k_2\rset} \proj{i}{i}\otimes \one_{d_2}\in \mathcal{U}(\C^{d_1}\otimes \C^{d_2})\, .
\]
Since $\rho$ is APPT we have that 
\[
U\rho U^\dagger = \sum^2_{s,t=1}\proj{k_s}{k_t}\otimes \sum_{lm} V_{sl}X_{k_s k_l}V^\dagger_{lt} + \sum_{ i,j\in \lset 1,\ldots ,d_1\rset\setminus\lset k_1 ,k_2\rset} \proj{i}{j}\otimes \rho_{ij}
\]
is PPT. This implies that the principal sub-block matrix 
\[
\sum^2_{s,t=1}\proj{k_s}{k_t}\otimes \sum_{lm} V_{sl}X_{k_s k_l}V^\dagger_{lt} = (W\otimes \one_{d_2})\tilde{U}\rho^{k_1k_2}\tilde{U}^\dagger (W\otimes \one_{d_2})^\dagger
\]
is PPT as well. Here $W:\C^2\ra \C^{d_1}$ denotes the isometry defined by $W\ket{s} = \ket{k_s}$ for $s=1,2$. Finally, since partial applications of isometries preserve positivity we have that $\tilde{U}\rho^{k_1k_2}\tilde{U}^\dagger$ is PPT. Since the unitary $\tilde{U}\in\mathcal{U}(\C^2\otimes \C^{d_2})$ was chosen arbitrarily we find that $\rho^{k_1k_2}$ is an (unnormalized) APPT state. By Theorem \cite[Theorem 1]{JohnstonAPPT2} this implies that $\rho^{k_1k_2}$ is ASEP and therefore in particular separable. This finishes the proof. 
\end{proof}

\section{Conclusion and outlook}
We have constructed explicit examples of PPT states on $\C^d\otimes \C^d$ with Schmidt numbers at least $\ceil{(d-1)/4}$. By modifying our construction slightly we even obtained states on $\C^d\otimes \C^d$ invariant under partial transpose with Schmidt numbers at least $\ceil{(d-3)/8}$. We continued by studying a family of random PPT states also achieving such a linear scaling in the local dimension, although without being able to specify the exact constants. Finally, we introduced a technique to derive upper bounds on the Schmidt number of general quantum states. Using this technique we show that for both PT-invariant and APPT states the Schmidt number cannot reach its maximal value. There are a number of open problems, and directions for further research to be explored from here:
\begin{itemize}
    \item What is the maximal Schmidt number of a PPT state (or a PT-invariant state) on $\C^{d_1}\otimes \C^{d_2}$ for general $d_1$ and $d_2$? We only know the answer when $\min(d_1,d_2)=2$ or when $d_1=d_2=3$. It is well-known that there exist entangled PPT states on $\C^{2}\otimes \C^{d}$ for any $d\geq 4$, a concrete example being the Tang-Horodecki state~\cite{Tang,HorodeckiState} (which is even PT-invariant~\cite{2xNseparability}) for $d=4$. For $d_1=d_2=3$ it has been shown~\cite{Yang16} that the Schmidt number of a PPT state can only reach $2$.   
    \item What can be said about the \emph{absolute Schmidt number} (i.e~\emph{Schmidt number from spectrum})? Also, can one estimate the size of bounded Schmidt number balls around the identity?
    \item What are the implications of positive partial transposition constraints on the Schmidt number in the multipartite setting? There exist states on $\C^3\otimes \C^3\otimes \C^3$ that are genuine multipartite entangled despite being PPT across every bipartition~\cite{GHLS}, but similar examples for higher Schmidt numbers are not known. A possible starting point could be the following questions: \emph{Is there a PPT state on a $\C^3\otimes \C^3\otimes \C^3$ with Schmidt number $3$ across every bipartition? Or what is the maximum genuine multipartite Schmidt number \cite{HuberVicente} possible for PPT states?} In this and related questions, an extension of positive map mixers \cite{GHLS} to $k$-positive map mixers (or of the corresponding positive maps \cite{Clivaz} to $k$-positive maps) might allow for an SDP/map based characterization. It would also be interesting to explore this question via probabilistic techniques in high dimensions.
    \item Despite being non-distillable, PPT states can be useful in other tasks. For instance, using \emph{data hiding} (see~\cite{DVLT,Matthews-VV,Ceci-Guillaume,LL-VV} for more details) PPT states arbitrarily close to \emph{private states}~\cite{HHHO,random-private} can be constructed. Such states have a high key rate. In multipartite systems PPT states can still yield metrological advantages \cite{vertesitoth}. Could one show that all PPT states with high Schmidt number are useful in some operational sense?
    \end{itemize}
    
    %Despite being non-distillable, PPT states can be useful in other tasks. For instance, PPT states which have the additional property of being \emph{data hiding} (see~\cite{DVLT,Matthews-VV,Ceci-Guillaume,LL-VV} for more details) can be interestingly exploited (on their own or through \emph{private states}~\cite{HHHO,random-private}). And in multipartite systems they can still yield metrological advantages \cite{vertesitoth}. Could one show, in a similar fashion, that PPT states which have the additional property of having a high Schmidt number constitute a key resource?

\emph{Acknowledgements}. We thank Matthias Christandl and Yu Yang for interesting discussions. Furthermore we gratefully acknowledge the hospitality of the Institut Henri Poincar\'e where part of this work has been conducted during the thematic trimester ``Analysis in Quantum Information Theory''. M.~Huber acknowledges funding from the Austrian Science Fund through the START project Y879-N27 and the joint Czech-Austrian project MultiQUEST I 3053-N27. L.~Lami acknowledges financial support from the European Research Council (ERC) under the Starting Grant GQCOP (Grant No. 637352). C.~Lancien is financially supported by the European Research Council through the grant agreement 648913, by the Spanish MINECO through the project MTM2014-54240-P, and by the Comunidad de Madrid through the project QUITEMAD+ S2013-ICE-2801. A.~M\"uller-Hermes acknowledges financial support from the European Research Council (ERC Grant Agreement no 337603), the Danish Council for Independent Research (Sapere Aude) and VILLUM FONDEN via the QMATH Centre of Excellence (Grant No. 10059). All authors contributed equally to this work.

\end{document}